\documentclass[aps,prb,twocolumn,showpacs,superscriptaddress]{revtex4-1}
\usepackage{graphicx}
\usepackage{dcolumn}
\usepackage{epstopdf}
\usepackage{bm}
\usepackage{color}
\usepackage[caption=false]{subfig}

\usepackage{amsfonts,amsthm,amsmath,amssymb}
\usepackage{hyperref}

\newcommand{\kn}{k_U}
\newcommand{\lmdn}{\lambda_U}
\newcommand{\En}{E_U}

\newcommand{\widebar}[1]{\overline{#1}}

\newcommand\numberthis{\refstepcounter{equation}\tag{\theequation}}

\newcommand{\todo}[1]{}

\newcommand{\sO}{\Psi_R}
\newcommand{\sB}{\widetilde{\Psi}_R}
\newcommand{\sP}{\widetilde{\Psi}_+}
\newcommand{\sM}{\widetilde{\Psi}_-}
\newcommand{\wpsi}{\widetilde{\psi}}
\newcommand{\psiR}{\psi_R}
\newcommand{\psiL}{\psi_L}

\newcommand{\tA}{\widetilde{A}}

\newcommand{\phiR}{\phi_R}
\newcommand{\phiL}{\phi_L}

\newcommand{\tr}{\mathop{\mathrm{tr}}}
\newcommand{\supp}{\mathop{\mathrm{supp}}}

\newcommand{\paperj}[1]{}

\newtheorem{lemma}{Lemma}
\newtheorem{proposition}{Proposition}

\theoremstyle{definition}

\begin{document}
\title{Chiral tunneling through generic one-dimensional potential barriers in bilayer graphene}

\author{V. Kleptsyn}
\email[]{victor.kleptsyn@univ-rennes1.fr}
\affiliation{CNRS, University of Rennes 1, Institute of Mathematical Research of Rennes, Rennes, France}
\author{A. Okunev}
\author{I. Schurov}
\affiliation{National Research University Higher School of Economics, Myasnitskaya 20, 101000 Moscow, Russia}
\author{D. Zubov}
\affiliation{Dept. of Differential Equations, Moscow State University, GSP-2, Leninskiye Gory, 119992 Moscow, Russia}
\author{M.~I. Katsnelson}
\affiliation{Radboud University, Institute for Molecules and Materials, Heijendaalseweg 135, 6525 AJ Nijmegen, The Netherlands}
\affiliation{Department of Theoretical Physics and Applied Mathematics, Ural Federal University, Mira str. 19, Ekaterinburg, 620002, Russia}
\date{\today}

\begin{abstract}
We study tunneling of charge carriers in single- and bilayer graphene. We propose an explanation for non-zero ``magic angles'' with 100\% transmission for the case of symmetric potential barrier, as well as for their almost-survival for slightly asymmetric barrier in the bilayer graphene known previously from numerical simulations. Most importantly, we demonstrate that these magic angles are {\it not} protected in the case of bilayer and give an explicit example of a barrier with very small electron transmission probability for {\it any} angles. This means that one can lock charge carriers by a p-n-p (or n-p-n) junction {\it without} opening energy gap. This creates new opportunities for the construction of graphene transistors.
\end{abstract}

\maketitle

\sloppy


\section{Introduction}
Klein tunneling, that is, transmission of massless Dirac fermions with a high probability through potential barriers, whatever broad and high they are \cite{Katsbook,Cetal09,KNG06,SRL08,SHG09,YK09,TRK12} is one of the key phenomena for graphene physics and technology. It protects conducting state of graphene with a high charge carrier mobility despite charge inhomogeneities \cite{Metal07}; at the same time, it does not allow to use the simplest construction of graphene transistor based on p-n-p (or \mbox{n-p-n}) junctions since such a device can never be locked \cite{KNG06}. As a result, some tricky ways should be used, for example, tunneling transistor with vertical geometry \cite{Betal12} where electrons tunnel between two sheets of graphene.

Full transmission for normally incident electron beam is symmetrically protected: since for massless Dirac electrons the propagation direction is intimately related to the direction of pseudospin, and the latter cannot be changed by an action of the electrostatic potential smooth at interatomic distances, the back scattering is completely forbidden \cite{Katsbook}; before the discovery of graphene, this was noticed as an explanation of stability of conducting channels in carbon nanotubes \cite{ANS98}. In the first calculation for the rectangular potential barrier in two-dimensional case \cite{KNG06}, additional nonzero ``magic angles'' of the incidence were found with also full transmission; these magic angles were also found for the parabolic barrier \cite{SRL08} and associated to Fabry-P\'erot resonances known in optics. These additional resonances are {\it not} universally protected. Semiclassical analysis \cite{TRK12,RTK13} has shown that these magic angles exist for symmetric potential barriers only whereas for generic one-dimensional barriers maxima of transmission corresponding to the Fabry-P\'erot conditions are suppressed; moreover, this suppression is exponential in a formal semiclassical smallness parameter.

Much less is known on the case of bilayer graphene where, in the simplest approximation, electron spectrum is massless but with parabolic touching instead of conical one and with nontrivial chiral properties of the charge carrier wave function \cite{Katsbook,Cetal09,Netal06,MF06}. In this case, for the normally incident electron beam the transmission probability is exponentially small \cite{KNG06}; this is a nice counterexample to attempts to relate Klein tunneling in single layer graphene ``just'' with the gapless character of the energy spectrum whereas the chiral properties of the wave functions are the most important. Existence of ``magic angles'' with full transmission in bilayer graphene has been demonstrated numerically for rectangular potential barrier \cite{KNG06} and for some smoothen shapes of the barrier~\cite{TRK12}. In the last paper, it was claimed that for asymmetric potentials, contrary to the case of single layer, magic angles survive. As we will see this statement is not quite accurate. Chiral effects in penetration of charge carriers through potential barriers in bilayer graphene has been studied also in Refs.~\onlinecite{GRL11} and~\onlinecite{DP13}. However, the issue of stability of magic angles in the bilayer graphene is still unclear. They have no obvious symmetry protection, like 100\% transmission at zero incident angle for the single layer graphene. Apart from theoretical interest this is a question of potentially great practical importance: if it would be possible to create a potential barrier with small enough electron transmission probability for any angles this would open a way to build a conventional transistor based on p-n-p (n-p-n)  junction in bilayer graphene without gap opening. In this paper we will show that this is, indeed, theoretically speaking, possible: the magic angles are not stable and electron transmission can be strongly suppressed by a proper choice of the shape of the barrier.

\section{Basic equations and the formulation of the problem}

\subsection{Single-layer graphene}\label{ssec:single-layer}
Quantum mechanics of charge carriers in graphene is governed by a massless Dirac equation \cite{Katsbook,Cetal09}:
\begin{equation}\label{eq:D}
\hat{H} \Psi = E \Psi,
\end{equation}
\begin{align*}
\hat{H}
& \, = V_F (\sigma_x \hat{p}_x +\sigma_y \hat{p}_y) + U(x,y) \\
& \,=  \left(
\begin{matrix}
U(x,y) & V_F  (\hat{p}_x - i \hat{p}_y) \\
V_F (\hat{p}_x + i \hat{p}_y) & U(x,y)
\end{matrix}
\right),\numberthis \label{eq:H-single}
\end{align*}
where $E$ is an energy of the stationary state, $\sigma_x$, $\sigma_y$ are Pauli matrices, $U(x,y)$ is the potential, $V_F\approx c/300$ is the Fermi velocity ($c$ is the velocity of light), $\hat{p}_x=-i\hbar \partial_x$, $\hat{p}_y=-i\hbar \partial_y$.

To study Klein tunneling, we restrict ourselves to the conventional case of a one-dimensional (i.e. given by a function depending on one variable $U(x,y)=U(x)$) potential barrier. A natural step is then to study the transmission and reflection probabilities as a function of the angle of incidence $\theta$. As mentioned in the Introduction, the transmission probability at $\theta=0$ is always 100\%, irrespective to the parameters of the potential. Numerical calculations show existence of 100\% transmission at additional \emph{magic angles} for symmetric potential barriers and just maxima at some angles for the asymmetric ones \cite{TRK12,RTK13}. To illustrate this, we show in Fig. \ref{fig:singlelayer} our computational results obtained for the same potentials of n-p-n junctions which were considered in Ref.\onlinecite{TRK12}, that coincide up  computer precision to those of~\cite[Fig. 6]{TRK12}. 

\begin{figure}[hbt]
    \begin{tabular}{ll}
        \subfloat[]{\includegraphics[width=4cm]{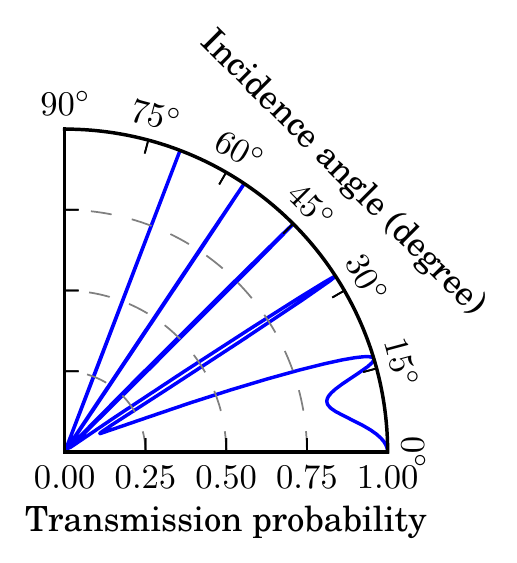}}&
        \subfloat[]{\includegraphics[width=4cm]{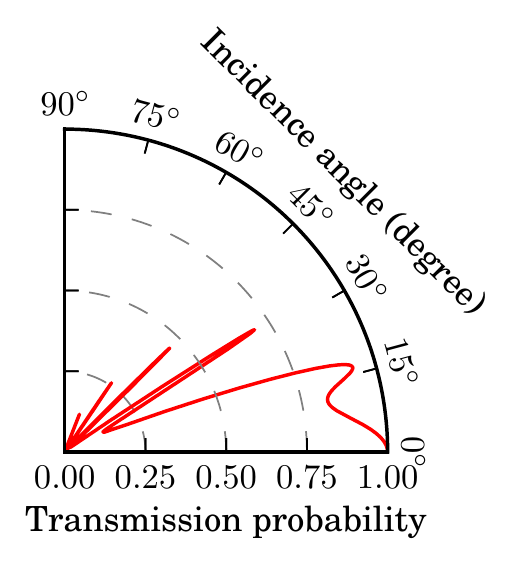}}\\
        \subfloat[]{\includegraphics[width=4cm]{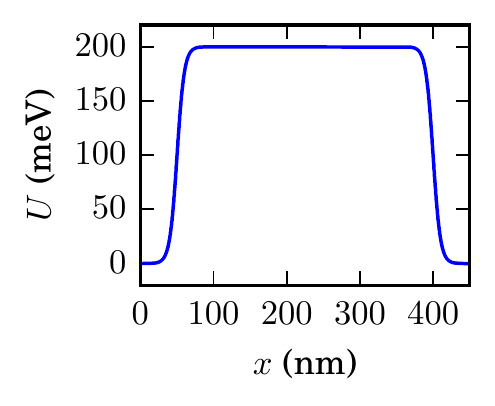}}&
        \subfloat[]{\includegraphics[width=4cm]{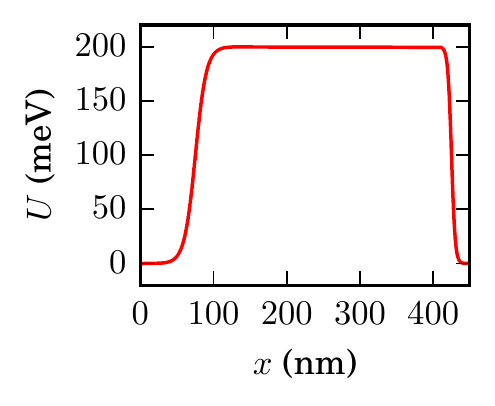}}
    \end{tabular}
    \caption{Numerical simulation of transmission probability for n-p-n junction in single layer graphene. 
       Energy of a particle 80 meV, height of the junction 200 meV. (a) Transmission probability for a 
       symmetric potential of barrier width $l_2=250$~nm, for which n-p and p-n regions have
       characteristic widths $l_1=l_3=100$~nm. (b) Transmission probability for an asymmetric potential
       with $l_1=150$~nm and $l_3=50$~nm. (c) and (d): the corresponding potentials.
        See~\cite[Fig.~6]{TRK12}.}\label{fig:singlelayer}
\end{figure}

From a general point of view, the existence of the magic angles for symmetric potentials and their disappearance for generic ones is an unexpected property of the Dirac equation (\ref{eq:H-single}). In Refs.\onlinecite{TRK12,RTK13} it was explained within a quite complicated semiclassical theory. It is interesting to study its origin \emph{per se}. Perfect transmission assumes that the reflection amplitude is zero; this means that a \emph{complex}-valued function of one \emph{real} argument, the angle of incidence, has nontrivial zeroes. We will give below a simple solution of this problem based only on symmetry properties of the Dirac Hamiltonian.

\subsection{Bilayer graphene}\label{s:intro-b}

The situation is different for the bilayer graphene. Following Ref.\onlinecite{KNG06}, we will deal with the simplest effective Hamiltonian describing chiral particles with parabolic dispersion:
\begin{equation}\label{eq:H-double}
\hat{H}= \left(
\begin{matrix}
U(x,y) &  (\hat{p}_x - i \hat{p}_y)^2/2m \\
 (\hat{p}_x + i \hat{p}_y)^2/2m & U(x,y)
\end{matrix}
\right),
\end{equation}
where $m\approx 0.031m_e$ is the effective mass of the electron in the bilayer graphene \cite{Metal11}, $m_e$ is the free electron mass. (This value is according to the latest experimental data; in early papers, the value $m = 0.054m_e$ was used.)

This Hamiltonian is not applied to the real bilayer graphene for very low energies ($E < 10$ meV) where the effects of trigonal warping are essential \cite{MF06,Metal11} and for high enough energies ($E > 200$ meV) where the transition to four-band picture \cite{DP13} is required. Anyway, Eq. (\ref{eq:H-double}) is a new type of wave equation different from both nonrelativistic Schr\"{o}dinger and Dirac equations and its study is by itself of significance for mathematical physics.

Contrary to the case of single-layer, there is no more Klein tunneling at zero angle of incidence and no clear symmetry properties which would protect full transmission at other angles, thus giving a hope to find a barrier that allows blocking of a current.

Numerical experiments for an n-p-n junction in bilayer~\cite{TRK12} have shown the
following:
\begin{itemize}
    \item nonzero magic angles are still present in a symmetric potential;
    \item moreover (and most surprizingly!) they seemed to survive when one passes
        to a non-symmetric potential, see~\cite[Fig. 7]{TRK12} and Fig.~\ref{fig:bilayer}(b) 
        with our computational results.
\end{itemize}
Both these effects thus were to be explained.

\begin{figure}[hbt]
    \begin{tabular}{cc}
        \subfloat[]{\includegraphics[width=4cm]{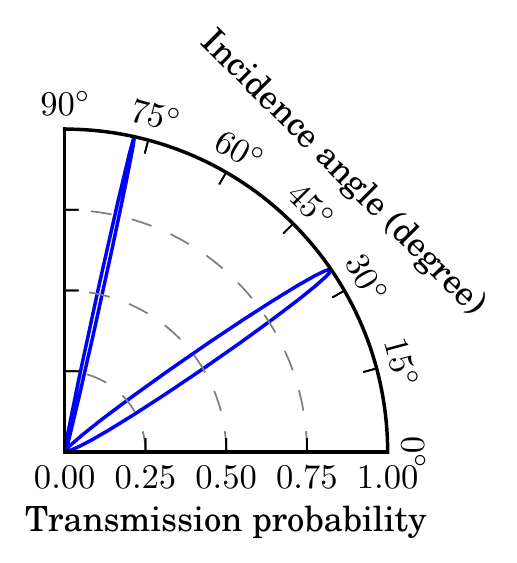}}&
        \subfloat[]{\includegraphics[width=4cm]{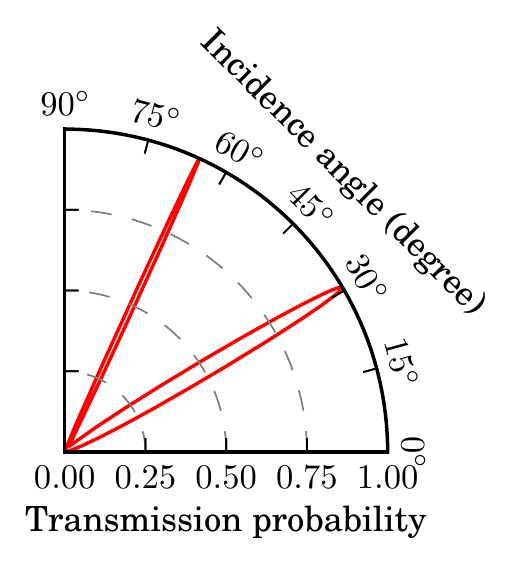}}\\
        \subfloat[]{\includegraphics[width=4cm]{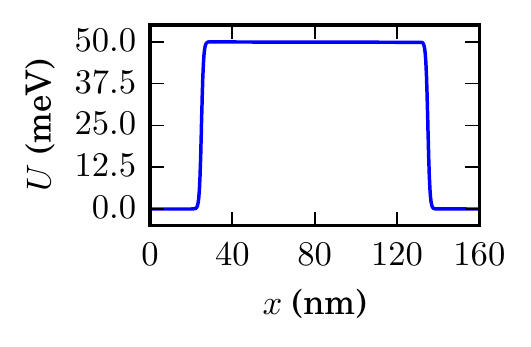}}&
        \subfloat[]{\includegraphics[width=4cm]{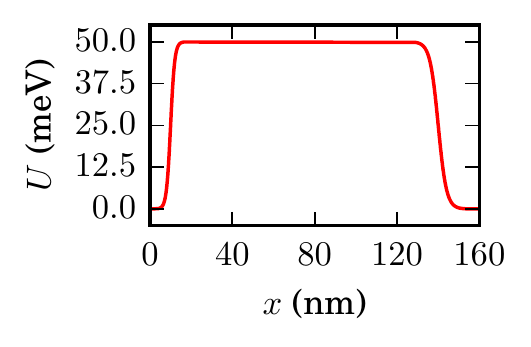}}
    \end{tabular}
\caption{Numerical simulations of a transmission through an n-p-n junction in bilayer graphene. Energy of a particle 17 meV, height of the junction 50 meV.
        (a) Transmission probability for a symmetric
        potential of barrier width $l_2=100$~nm, with n-p and p-n regions of
        characteristic width $l_1=l_3=10$~nm. (b) Transmission probability for an asymmetric potential
        with $l_1=20$~nm and $l_3=40$~nm. (c) and (d): the corresponding potentials.}
        \label{fig:bilayer}
\end{figure}

\subsection{The specific goals}
In this paper we are answering all the above questions as well as obtaining some other results. Namely, we
\begin{enumerate}
    \item Obtain an equivalent condition for a magic angle for a symmetric
        potential. This condition is given by one \emph{real} equation (with
        generically simple roots). This is done for both single-layer (see Sec.~\ref{s:single-layer})
        and bilayer (Sec.~\ref{s:bilayer}) graphene. In particular, the presence of magic
        angles for the symmetric n-p-n junctions (considered in Refs.~\onlinecite{TRK12,RTK13}) cannot be removed by a small
        perturbation of the potential within the class of symmetric ones.
    \item Explain (see Sec.~\ref{s:symmetric}) the seemingly magic angles for the nonsymmetric
        potentials (red line on Fig.~\ref{fig:bilayer}). It turns out that the minimal
        reflection probability is very small (of order $10^{-3}..10^{-4}$ for the potential
        studied in Ref.\cite{TRK12}) but is still nonzero (cf. Fig.~\ref{fig:zoomin}). The reflection/transmission probabilities
        for this potential turn out to be very close to the ones for a close
        symmetric potential (cf. Fig.~\ref{fig:sym-vs-asym}). The latter possesses exact magic angles, that
        become a seemingly magic for a non-symmetric one.
    \item Find (see Sec.~\ref{s:no-transmission}) an explicit example of a potential for which in some band of
        energies the transmission probability is less then $2\cdot 10^{-8}$ for any
        angle of incidence. Interestingly enough, such a potential
        can be taken to be symmetric, in particular, showing that magic angles
        are not obliged to be present for an arbitrary symmetric potential
        $U(x)$ for bilayer graphene.
    \item Provide (see Sec.~\ref{s:peaks}) an approximation for the transmission probability $p(\theta)$
        around a peak, showing that such a peak has (approximately) a standard
        Lorenz---Breit---Wigner form
        $$p(\theta)
       \approx \frac{p_0}{1+c(\theta-\theta_0)^2}$$
    \item Provide (see Sec.~\ref{s:algorithm}) a ``two-level'' method of finding the transmission probabilities
        in bilayer graphene that removes the exponential growth problem.
\end{enumerate}

\section{Symmetric case: true magic angles}\label{s:symmetric}
The key argument in both symmetric cases will be the following symmetry of the graphene equations:
\begin{equation}\label{eq:def-T}
T\colon
\left(\begin{matrix}
\Psi_1(x,y) \\ \Psi_2(x,y)
\end{matrix} \right)
\mapsto
\left( \begin{matrix}
\overline{\Psi}_2(-x,-y) \\ \overline{\Psi}_1(-x,-y)
\end{matrix} \right).
\end{equation}
It maps the solution of~\eqref{eq:D} for a potential $U(x,y)$ to the solution of~\eqref{eq:D} for the potential turned by 180$^{\circ}$ $U(-x,-y)$. Hence if $U(x,y)=U(x)$, as we will assume through the rest of the paper, and $U(-x)=U(x)$ (as we consider now the symmetric case), the map $T$ sends the solutions of~\eqref{eq:D} to the solutions of~\eqref{eq:D}. Finally, it is easy to notice that $T$ \emph{preserves} the direction of a flat wave.

It turns out that this symmetry of the problem reduces the number of independent equations. This is the same scenario that occurs, for instance, for the equation describing the Josephson junction~\cite{FGKS,KR14}.

\subsection{Single-layer case}\label{s:single-layer}
For the sake of mathematical completeness, let us state the problem formally. As
$U=U(x)$ does not depend on $y$, the
eigenfunctions can be tried in the form $\Psi(x,y)=\Psi(x,0)e^{i a y}$, where the corresponding eigenvalue of the operator $\hat{p}_y$ is equal to~$a\hbar$.
After denoting, by a slight abuse  of notation, $\Psi(x)=\Psi(x,0)$, the equation~\eqref{eq:D} becomes the \emph{ordinary} differential equation
\begin{equation}\label{eq:D-x-1}
\left(
\begin{matrix}
U(x)-E &   -iV_F\hbar (\partial_x +a) \\
-i V_F \hbar (\partial_x -a) & U(x)-E
\end{matrix}
\right)\Psi(x) =0,
\end{equation}
or, equivalently,
\begin{equation}\label{eq:D-x-1-explicit}
\partial_x \Psi(x)= \left(
\begin{matrix}
a & -i \frac{U(x)-E}{V_F\hbar} \\
 -i \frac{U(x)-E}{V_F\hbar} & -a
\end{matrix}
\right)\Psi(x).
\end{equation}

In the domain $U=0$ (that is, to the left or to the right of the barrier), the
solution of~\eqref{eq:D-x-1} is a linear combination of left- and right-going
waves $\psiL(x)=e^{-ikx}v_L$ and $\psiR(x)=e^{ikx}v_R$ respectively. Here $\pm ik,$ \, where $a^2+k^2=\left(\frac{E}{\hbar V_F}\right)^2$, are the eigenvalues of the right hand-side operator in~\eqref{eq:D-x-1-explicit}, and
$$
v_R= \left(\begin{matrix}
a+ik \\ i E/V_F\hbar
\end{matrix}
\right) \quad\text{and}\quad
v_L= \left(\begin{matrix}
a-ik \\ i E/V_F\hbar
\end{matrix}
\right)
$$
are the corresponding eigenvectors. For a wave of energy $E$ falling at the angle $\theta$ we thus have
$$
a=\frac{E}{\hbar V_F}\sin \theta, \quad  k=\frac{E}{\hbar V_F}\cos \theta;
$$
we can also rewrite the eigenvectors as
\begin{equation}\label{eq:vRLe}
v_R=\frac{E}{V_F\hbar}\left(\begin{matrix}
e^{i\theta} \\ i
\end{matrix}
\right) \quad\text{and}\quad
v_L= \frac{E}{V_F\hbar}\left(\begin{matrix}
e^{-i\theta} \\ i
\end{matrix}
\right).
\end{equation}

To find the transmission and reflection probabilities for a given $\theta$, one
then looks for the solution of~\eqref{eq:D-x-1} that is of the form

\begin{equation}\label{eq:t-and-r}
    \sO(x)=\begin{cases}
        \psiR(x)+r(\theta) \psiL(x) &\!\!\text{on the left of the barrier},\\
            t(\theta) \psiR(x) & \!\!\!\!\!\!\text{on the right of the barrier}.
    \end{cases}
\end{equation}
Here $r(\theta)$ and $t(\theta)$ are respectively reflection and transition
amplitutes, and the corresponding probabilities are the squares of their absolute
values.

In a general quantum mechanical setting, these amplitudes are general complex numbers with $|r(\theta)|^2+ |t(\theta)|^2=1$; the ``no-reflection'' condition $r=0$ on a \emph{complex} number $r$ cannot be satisfied in a generic one-parametric family.
It turns out, that in our particular situation of a symmetric potential, the numbers $r(\theta)$ and $t(\theta)$ satisfy an additional relation. Namely, we have the following
\begin{proposition}\label{p:qt}
The function $q(\theta)=i e^{-i\theta} \frac{r(\theta)}{t(\theta)}$, where $r(\theta)$ and $t(\theta)$ are defined by~\eqref{eq:t-and-r}, is real-valued for all values $\theta$.
\end{proposition}

Instead of the function $q(\theta)$, it is more convenient to consider its differently normalized version: the function $f(\theta)=\frac{q(\theta)}{\sqrt{1+q^2(\theta)}}$. This real-valued function, on the one hand, satisfies $|f(\theta)|=|r(\theta)|$, and hence its zeros are exactly the magic angles: see Fig.~\ref{fig:single-zero}. On the other hand, this real-valued function generically has simple zeroes (and as it takes values of both signs, its zeroes cannot be removed by a small perturbation).

\begin{figure}[hbt]
    \centering
    \includegraphics[width=8.5cm]{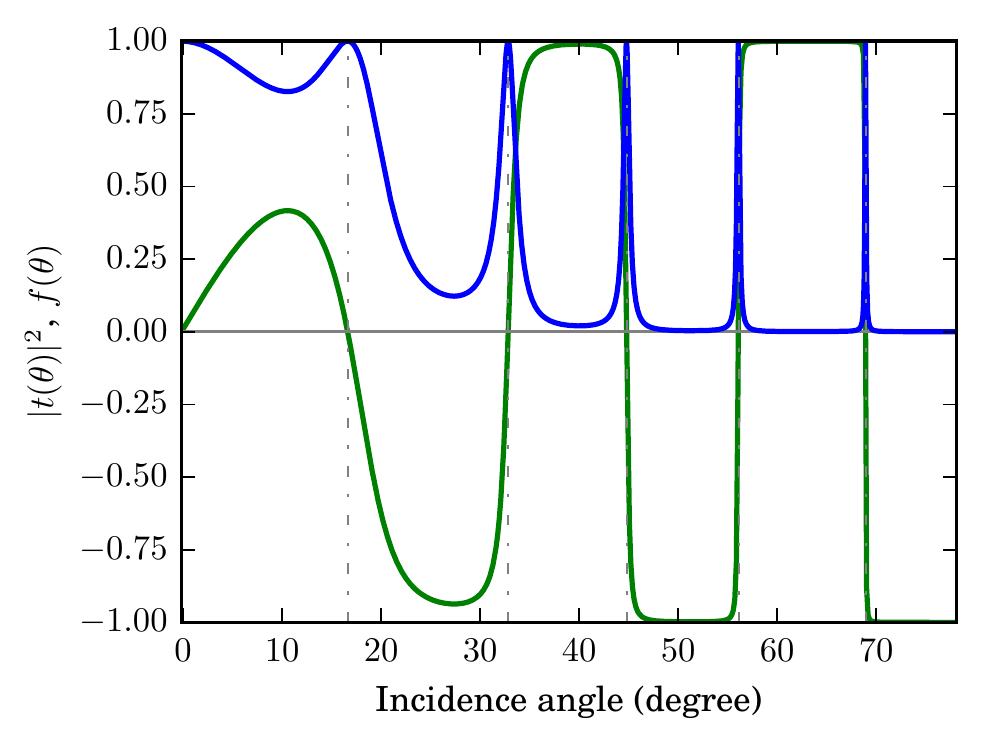}
    \caption{
        {Blue \paperj{(dotted) }line: the transmission probability for the symmetric potential from Fig.~\ref{fig:singlelayer}~(a,c).
        Green \paperj{(solid) }line: real-valued analytic function $f(\theta)$ such that $|r(\theta)|=|f(\theta)|$. Its zeroes correspond to magic angles.}
    }\label{fig:single-zero}
\end{figure}

\begin{proof}

First note that the map $T$, given by~\eqref{eq:def-T}, descends on the space of solutions
of~\eqref{eq:D-x-1}, preserving $a$:
\begin{equation}
    T\colon \begin{pmatrix}
        \Psi_1(x)\\
        \Psi_2(x)
    \end{pmatrix}
    \mapsto
    \begin{pmatrix}
        \widebar{\Psi}_2(-x)\\
        \widebar{\Psi}_1(-x)
    \end{pmatrix}.
\end{equation}

Second, remark that the flow~\eqref{eq:D-x-1-explicit} preserves the \emph{bilinear} antisymmetric form with the matrix
$$
    Q=\begin{pmatrix}
        0 & 1\\
        -1 & 0
    \end{pmatrix},
$$
that is, for any two solutions~$\psi_1, \psi_2$ the value $Q\left(\psi_1,\psi_2\right) = \psi_1^t(x) Q \psi_2(x)$ (``t'' means transposed) does not depend on $x$ (and thus $Q$ becomes a well-defined bilinear form on the space of solutions). The easiest way to see it is to say that the matrix of the flow~\eqref{eq:D-x-1-explicit} has zero trace, and hence the flow map from $x=x_1$ to $x=x_2$ has determinant one, while $\psi_1^t Q \psi_2 = \det (\psi_1 \, \psi_2)$.

It is also immediate to check that $Q(T\psi_1,T\psi_2)=\overline{Q(\psi_2,\psi_1)}$, that can be re-stated as
$$
Q(\psi_1,T\psi_2)=\overline{Q(\psi_2,T\psi_1)}.
$$
This remark easily implies the following observation
: for any solution $\psi$ of~\eqref{eq:D-x-1-explicit}, the value $Q(\psi, T\psi)$ is purely real.

Now, let us apply this remark to the solution~$\Psi=\frac{1}{t(\theta)}\sO$, where $\sO$ is of the form~\eqref{eq:t-and-r}. On one hand, $Q(\Psi,T\Psi)$ is a real number. On the other hand, calculating it at any point $x_0$ on the left of the barrier, we get
\begin{multline*}
Q(\Psi,T\Psi) = Q(\frac{1}{t(\theta)}\psiR(x_0)+\frac{r(\theta)}{t(\theta)}\psiL(x_0), (T \psiR) (x_0))=
\\ =\frac{r(\theta)}{t(\theta)} Q(v_L,Tv_R)=\frac{r(\theta)}{t(\theta)}\cdot (e^{- 2i\theta} -1) \frac{E^2}{V_F^2 \hbar^2},
\end{multline*}
where we have used~\eqref{eq:vRLe} for the last equality. Then,
$$
(e^{- 2i\theta} -1) \frac{E^2}{V_F^2 \hbar^2} = -e^{-i\theta} \cdot 2i \sin \theta \cdot \frac{E^2}{V_F^2 \hbar^2} = -i e^{-i\theta} \cdot \frac{2 aE}{V_F \hbar},
$$
and hence
$$
Q(\Psi,T\Psi) = - q(\theta) \cdot \frac{2 aE}{V_F \hbar}.
$$
The left hand side of the last equality is real, and thus so is $q(\theta)$. 

\end{proof}

To conclude this paragraph, let us restate the condition for an angle being magic, $r(\theta)=0$, in two different ways. First, due to the proof of Prop.~\ref{p:qt} it is equivalent to $Q(\sO,T\sO)=0$; calculation of this form at $x=0$ gives us
$$
Q(\sO,T\sO)=|\Psi_{R,1}(0)|^2-|\Psi_{R,2}(0)|^2.
$$
Second, the form $Q$ is non-degenerate, so $Q(\psi_1,\psi_2)=0$ if and only if $\psi_1$ and $\psi_2$ are proportional. We finally get
\begin{proposition}
    The angle $\theta$ is magic if and only if
    $$|\Psi_{R,1}(0)|=|\Psi_{R,2}(0)|,$$
    and if and only if $T \sO$ is proportional to~$\sO$.
\end{proposition}

\subsection{Bilayer case}\label{s:bilayer}
In the same way as in the single-layer case, we are considering the solutions of
the form $\Psi(x,y)=\Psi(x) e^{iay}$. Then equation~\eqref{eq:D} becomes
\begin{equation}\label{eq:D-x-2}
\left(
\begin{matrix}
U(x)-E &   -\frac{{\hbar}^2}{2m}(\partial_x+a)^2 \\
-\frac{{\hbar}^2}{2m}(\partial_x -a)^2 & U(x)-E
\end{matrix}
\right)\Psi(x) =0
\end{equation}
or, introducing new variables
\begin{equation}\label{eq:realcoords}
    \begin{matrix}
\tilde{U}(x) = \frac{2m}{{\hbar}^2}(U(x)-E), & \Phi_1(x) = (\partial_x -a)\Psi_1(x),
\\
& \Phi_2(x) = (\partial_x +a)\Psi_2(x),
\end{matrix}
\end{equation}
we can reduce~\eqref{eq:D-x-2} to the following equation:
\begin{equation}\label{eq:D-x-4x4}
\partial_x {\tilde{\Psi}} = \left(
\begin{matrix}
a & 1 & 0 & 0 \\
0 & a & \tilde{U}(x) & 0 \\
0 & 0 & -a & 1 \\
\tilde{U}(x) & 0 & 0 & -a
\end{matrix} \right) \tilde{\Psi},
\end{equation}
on a complex $4$-vector $${\tilde{\Psi}}(x) = (\Psi_1(x), \Phi_1(x),\Psi_2(x), \Phi_2(x))^{t}.$$

In the same way as for the single-layer case, in the domain
$U=0$, equation~\eqref{eq:D-x-4x4} has the solutions of the form (for $E>0$, otherwise some signs will be different)
$$\psiR(x)=e^{ikx} v_R, \quad \psiL(x)=e^{-ikx} v_L,$$
where
\begin{equation}\label{eq:vRL}
v_R=\left(
\begin{matrix}
a+ ik \\ -\frac{2m}{{\hbar}^2} E 
\\
-(a- ik) \\ -\frac{2m}{{\hbar}^2} E
\end{matrix}
\right)
\quad \text{ and } v_L=\left(
\begin{matrix}
a- ik \\ -\frac{2m}{{\hbar}^2} E  \\
-(a+ ik) \\ -\frac{2m}{{\hbar}^2} E
\end{matrix}
\right)
\end{equation}
and $a^2+ k^2 = \frac{2m}{{\hbar}^2} E$. The angle of incidence $\theta$ is now related to $a$ by $a=\sqrt{\frac{2m}{{\hbar}^2}E} \sin \theta$ and
$k=\sqrt{\frac{2m}{{\hbar}^2} E} \cos \theta$.

Though, the equation~\eqref{eq:D-x-4x4} has also in the domain $U=0$ the solutions of the form
$$
\psi_+(x)=e^{\lambda x}v_+,\quad \psi_-(x)=e^{-\lambda x} v_-,
$$
where
\begin{equation}\label{eq:vpm}
v_+=\left(
\begin{matrix}
a+ \lambda \\ \frac{2m}{{\hbar}^2}E \\
a- \lambda \\ -\frac{2m}{{\hbar}^2}E
\end{matrix}
\right), \quad v_-=\left(
\begin{matrix}
a- \lambda \\ \frac{2m}{{\hbar}^2}E \\
a+ \lambda \\ -\frac{2m}{{\hbar}^2}E
\end{matrix}
\right), \quad \lambda=\sqrt{a^2+\frac{2m}{{\hbar}^2}E}.
\end{equation}

To find the transmission and reflection probabilities one now has to consider
the solution of the form

\begin{equation}\label{eq:t-and-r-bilayer}
    \sB(x)=\begin{cases}
        \psiR(x)+r(\theta) \psiL(x)+\alpha_1 \psi_+(x),& x<-x_0\\
        t(\theta) \psiR(x)+\alpha_2 \psi_-(x), & x>x_0
    \end{cases}
\end{equation}
where $\supp U\subset [-x_0,x_0]$.

We then have the following
\begin{proposition}\label{p:bilayer}
The function $q(\theta)=i \frac{r(\theta)}{t(\theta)}$, where $r(\theta)$ and $t(\theta)$ are defined by~\eqref{eq:t-and-r-bilayer}, is real-valued for all values~$\theta$.
\end{proposition}

Once again, instead of the function $q(\theta)$ we can consider the function $f(\theta)=\frac{q(\theta)}{\sqrt{1+q^2(\theta)}}$, which is more convenient due to the relation $|r(\theta)|=|f(\theta)|$. Zeroes of $f(\theta)$ correspond to the magic angles, and are generically simple (see Fig.~\ref{fig:bilayer-zero}).

\begin{figure}[hbt]
    \centering
    \includegraphics[width=8.5cm]{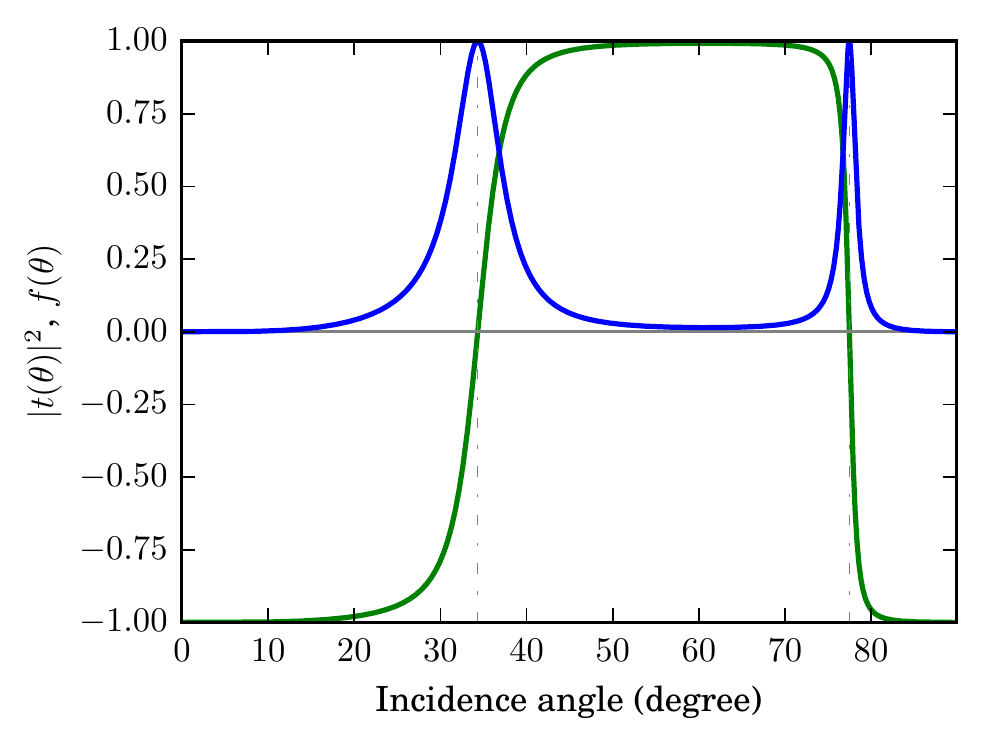}
    \caption{
        {   Blue \paperj{(dotted)} line: the transmission probability for the symmetric potential from Fig.~\ref{fig:bilayer}~(a,c).
        Green \paperj{(solid)} line: 
         real-valued analytic function $f(\theta)$ such that $|r(\theta)|=|f(\theta)|$. Again, its zeroes correspond to magic angles.
        } 
    }\label{fig:bilayer-zero}
\end{figure}

\begin{proof}

In the same way as in the single-layer case, the map $T$ descends on the
space of solutions of~\eqref{eq:D-x-4x4}, preserving $a$, and becomes:
\begin{equation}
        T\colon
        \begin{pmatrix}
            \Psi_1(x)\\
            \Phi_1(x)\\
            \Psi_2(x)\\
            \Phi_2(x)
        \end{pmatrix}
        \mapsto
        \begin{pmatrix}
            \overline{\Psi}_2(-x)\\
            -\overline{\Phi}_2(-x)\\
            \overline{\Psi}_1(-x)\\
            -\overline{\Phi}_1(-x)
        \end{pmatrix}.
\end{equation}

Also in the same way as before, we note that the flow of the equation~\eqref{eq:D-x-4x4} preserves a bilinear antisymmetric form. Namely, note first that is preserves a  \emph{sesquilinear} form with the matrix
$$Q=\begin{pmatrix}
    0 & 0 & 0 & 1\\
    0 & 0 & -1 & 0\\
    0 & 1 & 0 & 0\\
    -1 & 0 & 0 & 0
\end{pmatrix}.$$
This can be checked by explicit computation; though, there is a simple physical
interpretation for it. Namely, $\frac{1}{i}\Phi^*(x_1)Q\Phi(x_1)$ measures the current density
through the section $x=x_1$ (see, e.g.~\cite[Eqs. (136)--(137)]{TRK12}), so the conservation of the form is merely the
conservation of number of particles.

At the same time, both matrices of the flow and of
the form $Q$ are purely real. Hence, the flow also preserves an anisymmetric
\emph{bilinear} form with the same matrix~$Q$. Hence~$Q(\cdot,\cdot)$ is a
well-defined antisymmetric form on the space of solutions of~\eqref{eq:D-x-4x4}:
the value $(\widetilde{\Psi}_{(1)}(x))^{t} Q \widetilde{\Psi}_{(2)}(x)$
does not depend on the choice of the point~$x$.

Finally, it is easy to see that we have the following
\begin{lemma}
For any two solutions $\widetilde{\psi}_1, \widetilde{\psi}_2$ of~\eqref{eq:D-x-4x4} one has
$$
Q(\widetilde{\psi}_1,T\widetilde{\psi}_2)=\overline{Q(\widetilde{\psi}_2,T\widetilde{\psi}_1)}.
$$
In particular, $Q(\widetilde{\psi}_1,T\widetilde{\psi}_1)$ is always real.

For the eigenvectors of the matrix of the flow~\eqref{eq:D-x-4x4} at zero potential, the only pairs giving nonzero product $Q(\widetilde{\psi}_1,T\widetilde{\psi}_2)$ are
$$
Q(v_R, Tv_L) = - Q(v_L,Tv_R) = 4ik\cdot\frac{2mE}{{\hbar}^2},
$$
and
$$
Q(v_+,Tv_+)=-Q(v_-,Tv_-) = - 4\lambda\cdot \frac{2mE}{{\hbar}^2}.
$$
\end{lemma}
\begin{proof}
The first part is immediate; the equalities in the second part can be checked by an explicit computation. Finally, to check that all the other likewise products vanish, note, that
$$
T(v_R)= -v_R, \,\, T(v_L)=- v_L, \,\, T(v_+)= v_-, \,\, T(v_-)= v_+,
$$
and that due to the invariance of $Q$ by the flow the $Q$-product $Q(v_1,v_2)$ on two eigenvectors $v_1,v_2$ can be non-zero only if the sum of corresponding eigenvalues vanishes.
\end{proof}

Take now
\begin{align*}
\widetilde{\Psi}(x)&:=\frac{1}{t(\theta)} \sB(x)
\\ & =\begin{cases}
        \frac{1}{t(\theta)}\psiR(x)+\frac{r(\theta)}{t(\theta)} \psiL(x)+\frac{\alpha_1}{t(\theta)} \psi_+(x),& x<-x_0,\\
         \psiR(x)+\frac{\alpha_2}{t(\theta)} \psi_-(x), & x>x_0.
    \end{cases}
\end{align*}
Then, $Q(\widetilde{\Psi}, T(\widetilde{\Psi}))$ is a real number. Note that at $x<x_0$, we have $(T\widetilde{\Psi})(x)=-\psiR(x)+\overline{(\alpha_2/t(\theta))} \psi_+(x))$, and hence
$$
Q(\widetilde{\Psi}, T(\widetilde{\Psi})) = \frac{r(\theta)}{t(\theta)} Q(v_L,Tv_R) = -4ik\cdot \frac{2mE}{{\hbar}^2} \cdot \frac{r(\theta)}{t(\theta)}
$$
(The expression in the middle is the only term that does not vanish.) As the expression in the left hand side is real, so is $q(\theta)=i \frac{r(\theta)}{t(\theta)}$.
\end{proof}


\section{Transmission-blocking example}\label{s:no-transmission}

The above arguments show that in the symmetric case, if magic angles were present for some potential $U$, they usually can not be removed by its \emph{small} perturbation. Though, these arguments do not imply that the magic angles should be present for \emph{any} symmetric potential: indeed, a real function $q(\theta)$ generically is not \emph{obliged} to have real roots.

Indeed, one can construct an example of a symmetric barrier, the probability of transmission through which is quite small. The potential that we construct is a sufficiently quickly oscillating one, being a series of four n-p-n barriers. Namely, fix the height of barriers $U_0=50 meV$ and the n-p and p-n junction widths $l_1=l_3=\dots=l_{15}=10nm$ and the pairwise distances and widths $l_2=l_4=\dots=l_{14}=10nm$. Then, take
$$
x_0=0, \quad x_{i+1}=x_i+l_i, \quad i=1,\dots,15
$$
and define
\begin{equation}\label{eq:four-pot}
U(x)=\left\{\begin{array}{l}
0, \quad \,\,\,\text{if $x<x_0$ or $x>x_{15}$ or $x\in [x_{4i-1},x_{4i}]$},\\
U_0, \quad \text{if $x\in [x_{4i+1},x_{4i+2}]$},\\
U_0 \cdot \frac{1}{2} \left( 1+ \tanh (10 (x-\frac{x_{4i}+x_{4i+1}}{2})/l_{4i}) \right),
\\ \qquad \,\,\,\text{if $x\in[x_{4i},x_{4i+1}]$},\\
U_0 \cdot \frac{1}{2} \left( 1- \tanh (10 (x-\frac{x_{4i+2}+x_{4i+3}}{2})/l_{4i+2}) \right),
\\ \qquad \,\,\,\text{if $x\in[x_{4i+2},x_{4i+3}]$}
\end{array}
\right.
\end{equation}
(see Fig.~\ref{fig:potential}; compare with \cite[Eq. (135)]{TRK12}).

\begin{figure}[hbt]
    \centering
    \includegraphics[width=8.5cm]{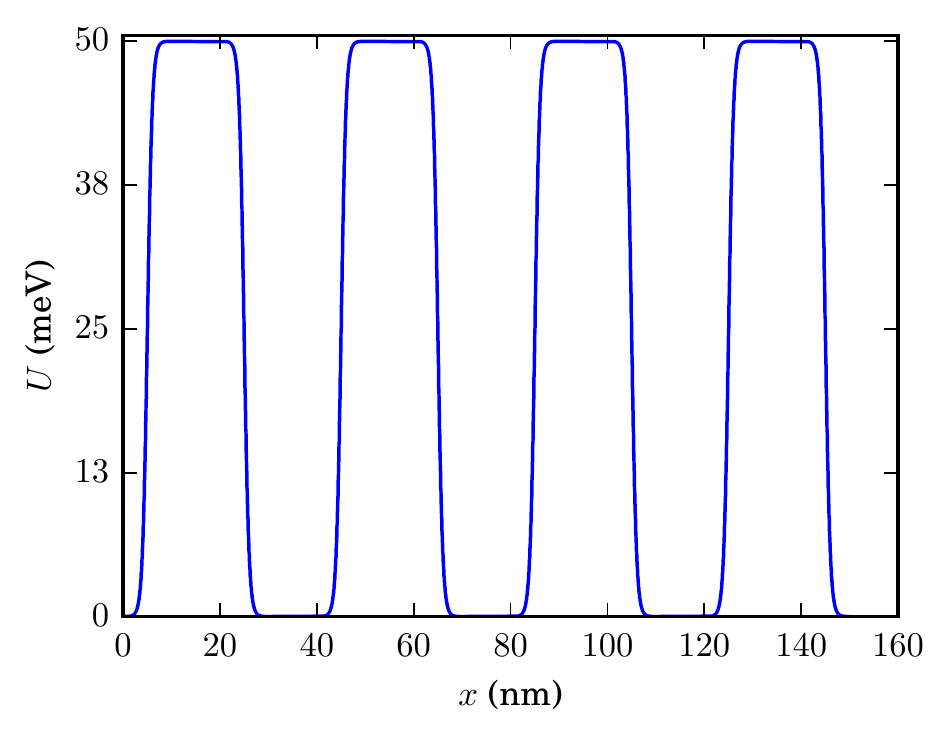}
    \caption{
        {Fast-oscillating potential $U(x)$.}
    }\label{fig:potential}
\end{figure}

The fast oscillations of the potential $U(x)$ prevent the appearances of ``resonances'' between the junctions, and allow to block the transmission in a band of energies sufficiently close to~$U_0/2$. Namely, in the energy band from 20 meV to 30 meV, the transmission probability $p(\theta)=|t(\theta)|^2$ for any angle of incidence~$\theta$ does not exceed $2\cdot 10^{-8}$\todo{Possibly: provide two examples of bands; possibly: draw examples with 3 and with 4 jumps.}: see Fig.~\ref{fig:p-max}.

\begin{figure}[hbt]
    \centering
    \begin{tabular}{c}
        \subfloat[]{\includegraphics[width=8.2cm]{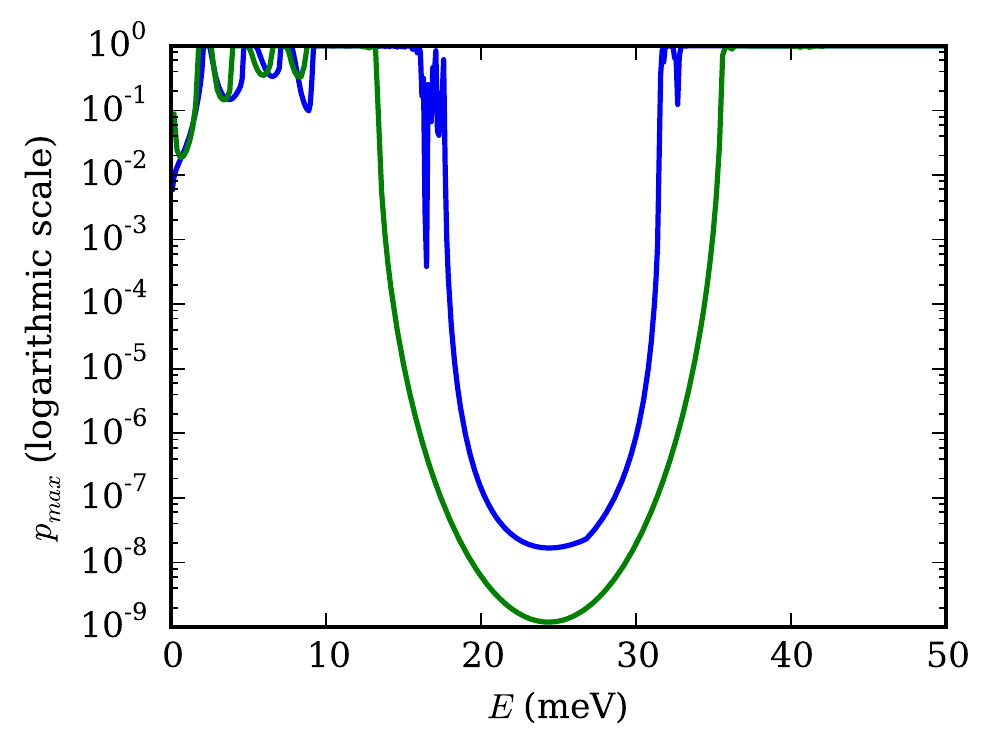}}\\
        \subfloat[]{\includegraphics[width=8.5cm]{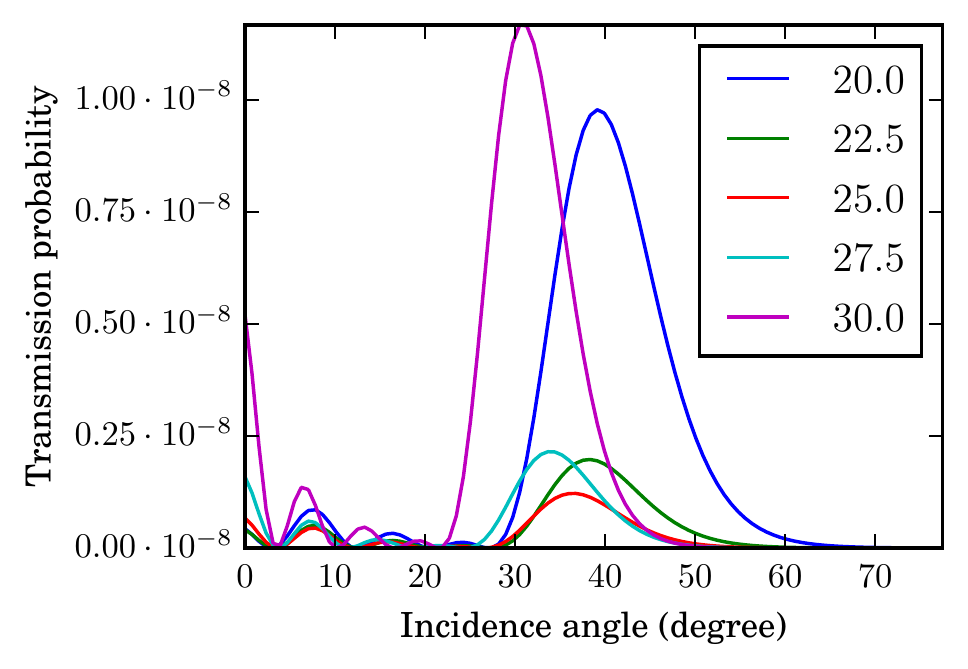}}
    \end{tabular}
    \caption{
        {(a) Green line: maximal transmission probability $p_{max}=\max_{\theta} p(\theta)$ for the
        potential~\eqref{eq:four-pot} as a function of the given energy~$E$ of the wave. Note that in the band E=20..30~meV 
        the transmission probability does not exceed~$2\cdot 10^{-8}$.
        Blue line: maximal transmission probability for the cosine-oscillating potential~\eqref{eq:four-cos}. Note, that in the same energy band the transmission probability 	is still very small (it does not exceed $10^{-6}$). (b) transmission probability for the
        potential~\eqref{eq:four-pot} as a function of the angle $\theta$ for energies $E=20,22.5,25,27.5,30$~meV. }
    }\label{fig:p-max}
\end{figure}

Further increase in number of barriers allows further reducing of transmission probability in this band (or of slight widening of the band where a given upper estimate on $|t(\theta)|^2$ holds). 
We will discuss in the next section a reason of why a blocking potential should be looked for among fast-oscillating ones.
We conclude this section by noticing that other fast-oscillating potentials exhibit similar behavior. For instance, considering the potential
\begin{equation}\label{eq:four-cos}
U(x)= \begin{cases}
U_0\cdot \frac{1}{2} (1-\cos (\pi x/ 2 l_1)),  & x\in [0, 16l_1]\\
0 & \text{otherwise},
\end{cases}
\end{equation}
we get a similar behavior for the maximal transmission probability: see Fig.~\ref{fig:p-max}, top.

\section{Bilayer graphene: generic one-dimensional potential}
\subsection{Approximate magic angles}\label{s:approximate}

We start by considering the example that was studied
in~Ref.\onlinecite{TRK12}: an asymmetric n-p-n junction. First, note that the
corresponding angles are not \emph{exactly} magic: though they seem to be such,
precise computations show that the local minimal values of $|r(\theta)|$ in these cases
are $\approx{}0.029$ and $\approx{}0.053$ at the angles $\theta\approx{}30^{\circ}$  and $\theta\approx{}65^{\circ}$ respectively
(and hence the local maxima of the transmission probability $|t(\theta)|^2$ are respectively $\approx{}0.9991$ and $\approx{}0.997$); see Fig.~\ref{fig:zoomin}.

\begin{figure}[hbt]
    \centering
    \begin{tabular}{cc}
        \subfloat[]{\includegraphics[width=4.3cm]{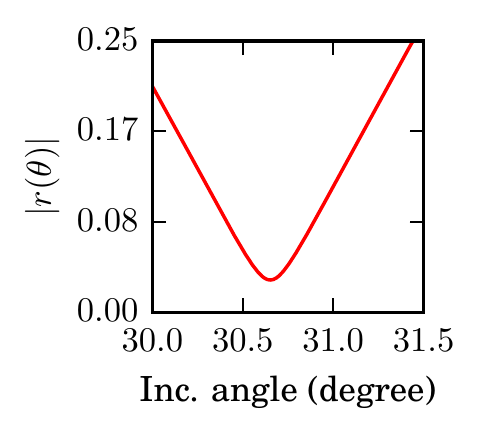}} &  
        \subfloat[]{\includegraphics[width=4.3cm]{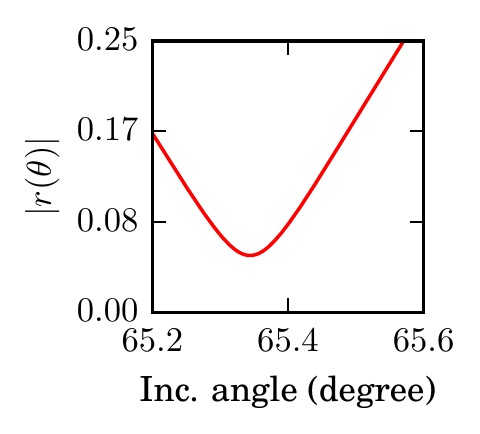}}
    \end{tabular}
    \caption{
        {Zoom in: absolute value of the reflection amplitude $|r(\theta)|$ for n-p-n junction in bilayer graphene near magic angles.
        Settings are the same as in Figure~\ref{fig:bilayer}~(b,d).}
    }\label{fig:zoomin}
\end{figure}

Thus the correct mathematical question is not to explain the \emph{exact} equality but to
explain why the minimal value of $|r(\theta)|$ is so small. 

The first idea here would be to compare this asymmetric potential to a close symmetric one. Namely, 
in addition to the asymmetric potential with junction widths
$l_1=20$~nm and $l_3=40$~nm and flat p-part of width $l_2=100$~nm, consider a symmetric one with the junction widths $l_1=l_3=20$~nm and flat p-part of width $l_2=110$~nm (so that the middle of the p-n junction does not move). It turns out that the transition probabilities for these barriers are quite close to each other (see Fig.~\ref{fig:sym-vs-asym}, red and green lines). 
At the same time, latter potential possesses \emph{exact} magic angles due to the same arguments as in Sec.~\ref{s:bilayer}. Given this, one could 
argue that the reason for almost-magic angles is just that the solutions of two equations, the one for the symmetric potential and for the asymmetric 
one, are sufficiently close to each other.

However, a closer examination of the graph on Fig.~\ref{fig:sym-vs-asym} shows that this argument alone cannot be a satisfactory explanation. 
Indeed, one notices that the peaks are shifted with respect to each other, so that the value of one of the functions at the maximum of the other one is quite far from~1, much farther than its maximal value. Hence, additional arguments are required for an explanation here.

In fact, note that altering \emph{both} junction widths, that is, considering the symmetric potential with $l_1=l_3=30$~nm, $l_2=100$~nm, one gets a much better approximation for the transmission probability: see Fig.~\ref{fig:sym-vs-asym}, red and blue lines. Though, \emph{a priori} it is not clear why such an approximation (contrary to the one-side modification) gives so precisely the approximately magic angles. We will explain it at the end of this subsection.

To provide a complete explanation for the almost-magic angles effect, we will approximate the problem of crossing of an n-p-n barrier as a sequence of two independent crossings, of an n-p and of a p-n barriers respectively. The error in such an approximation will be almost neglectable, thus reducing our question to the study of individual crossings. Finally, an additional effect, appearing in the bilayer graphene (contrary to the single-layer) is in the core of the explanation here.

\begin{figure}[hbt]
    \centering
\includegraphics[width=8cm]{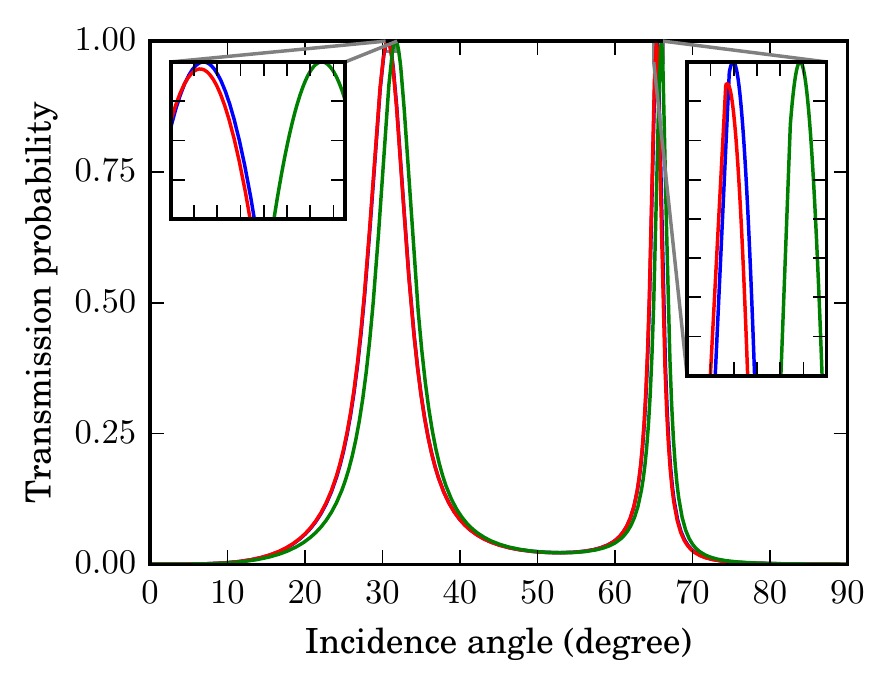}
    \caption{
        {
        Transmission probability for n-p-n junction in bilayer graphene for asymmetric 20-100-40~nm (red \paperj{solid }lines), symmetric 20-110-20~nm (green lines) and 30-100-30~nm (blue \paperj{dotted }lines) barriers.}
    }\label{fig:sym-vs-asym}
\end{figure}

Namely, under the width $l_2$ flat part of the barrier, the equation~\eqref{eq:D-x-4x4} again has constant coefficients, and can be interpreted as zero-potential equation for a wave of energy $\En:=E-U_0$. Provided that $|\En|>|a|=|E \sin \theta|$, the wave-type solutions of~\eqref{eq:D-x-4x4} in this domain can be written in the form
$$\phiR(x)=e^{i\kn x} v_R', \quad \phiL(x)=e^{-i\kn x} v_L',$$
where $\kn =\sqrt{\En^2 - a^2}$. Here $v_R'$ and $v_L'$ are the corresponding eigenvectors of the matrix of the system, given by~\eqref{eq:vRL} for the energy~$\En$, up to some sign changes due to the inequality $\En<0$:  
\begin{equation}\label{eq:vRL-prim}
v_R'=\left(
\begin{matrix}
a+ i\kn \\ \frac{2m}{{\hbar}^2} \En \\
a- i\kn \\ -\frac{2m}{{\hbar}^2} \En
\end{matrix}
\right)
\quad \text{ and } v_L'=\left(
\begin{matrix}
a- i\kn \\ \frac{2m}{{\hbar}^2} \En  \\
a+ i\kn \\ -\frac{2m}{{\hbar}^2} \En
\end{matrix}
\right)
\end{equation}
Also, as earlier, one also finds in this domain solutions of the form
$$\phi_+(x)=e^{\lmdn x} v_+', \quad \phi_-(x)=e^{-\lmdn x} v_-',$$
where $\lmdn=\sqrt{\En^2 + a^2}$, and $v'_{\pm}$ are the associated eigenvectors (again with a sign change with respect to~\eqref{eq:vpm}):
\begin{equation}\label{eq:vpm-prim}
v_+=\left(
\begin{matrix}
a+ \lmdn \\ -\frac{2m}{{\hbar}^2}\En \\
-(a- \lmdn) \\ -\frac{2m}{{\hbar}^2}\En
\end{matrix}
\right), \quad v_-=\left(
\begin{matrix}
a- \lmdn \\ -\frac{2m}{{\hbar}^2}\En \\
-(a+ \lmdn) \\ -\frac{2m}{{\hbar}^2}\En
\end{matrix}
\right).
\end{equation}

Now, related to the problem of describing an n-p-n barrier, one can consider the problem of describing an n-p transmission, given by a potential $U_{n-p}(x)$ that is identically~$0$ on the left of the barrier and identically $U_0$ on the right of it. For a barrier of characteristic width~$l$, analogously to the n-p-n barrier, we can take 
\begin{equation}\label{eq:n-p-potential}
U(x)=\left\{\begin{array}{l}
0, \quad \,\,\,\text{if $x<0$},\\
U_0, \quad \text{if $x>l$},\\
U_0 \cdot \frac{1}{2} \left( 1+ \tanh (10 (x-\frac{l_1}{2})/l) \right),
\\ \qquad \,\,\,\text{if $x\in[0,l]$},
\end{array}
\right.
\end{equation}
(compare with \cite[Eq. (135)]{TRK12}). A physically meaningful solution then should be of the form 
\begin{equation}\label{eq:sol-n-p}
\Psi_{n-p}(x)= \begin{cases}
        a_1\psiR(x)+a_2 \psiL(x)+a_3 \psi_+(x),& \\
         \qquad \text{ to the left of the barrier}, & \\
        a_4\phiR(x)+a_5 \phiL(x)+a_6 \phi_-(x), & \\
         \qquad \text{ to the right of the barrier}. &
    \end{cases}
\end{equation}

The coefficients $a_1,a_2$ of the wave component to the left of the barrier and the coefficients $a_4, a_5$ to the right of it are then related by a transmission matrix $A_{n-p}=A_{n-p}(l,\theta)$:
$$
\left( \begin{array}{c}
a_1 \\
a_2
\end{array} \right)= A_{n-p} \left( \begin{array}{c}
a_4 \\
a_5
\end{array} \right).
$$
Note now, that for any coordinates on the the 2-dimensional space of physical solutions, the restriction of the bilinear antisymmetric form $Q$ on this plane is proportional to the determinant (area) form in these coordinates. Hence, considering two different systems of coordinates $(a_1,a_2)$ and $(a_4,a_5)$, we see that the coefficient of proportionality is equal to $Q(\psiR,\psiL)=4ik\cdot\frac{2m}{\hbar^2}E$ and $Q(\phiR,\phiL)=4i\kn\cdot\frac{2m}{\hbar^2}\En$ respectively, and thus the determinant of the matrix $A_{n-p}$, relating these coordinates, is equal to
$$
\det A_{n-p} = \frac{4i\kn\cdot\frac{2m}{\hbar^2}\En}{4ik\cdot\frac{2m}{\hbar^2}E}=\frac{\kn\En}{kE}.
$$

Now, pass from the solutions $\psiR, \psiL$ in $U=0$ domain and from $\phiR,\phiL$ in the $U=U_0$ one to the properly normalized ``sine-cosine'' solutions 
$$
\psi_{\cos}:=\frac{\sqrt{kE}}{2}(\psiR+\psiL), \quad \psi_{\sin}:=\frac{\sqrt{kE}}{2i}(\psiR-\psiL)
$$ 
in the domain $U=0$ and to 
$$
\phi_{\cos}:=\frac{\sqrt{-\kn\En}}{2}(\phiR+\phiL), \quad \phi_{\sin}:=-\frac{\sqrt{-\kn\En}}{2i}(\phiR-\phiL)
$$ 
in the domain $U=U_0$. The advantage of these solutions is that as they are purely real, the transmission matrix in these coordinates
\begin{equation}\label{eq:np-tilde}
\tA_{n-p} = \frac{\sqrt{kE}}{\sqrt{-\kn\En}} \left(\begin{smallmatrix} \frac{1}{2} & \frac{1}{2} \\ \frac{1}{2i} & -\frac{1}{2i} 
\end{smallmatrix} \right) A_{n-p}  \left(\begin{smallmatrix} \frac{1}{2} & \frac{1}{2} \\ -\frac{1}{2i} & \frac{1}{2i} 
\end{smallmatrix} \right)^{-1}
\end{equation}
is also purely real, and due to the choice of the normalization it is of determinant~$1$.

Recall now that a real area-preserving matrix $\widetilde{A}$ admits singular value decomposition: it can be represented as a product of a rotation matrix, a diagonal matrix $\left(\begin{smallmatrix} \mu & 0 \\ 0 & \mu^{-1} 
\end{smallmatrix} \right)$, and another rotation matrix. The reflection and transmission coefficients, associated to such a matrix then satisfy 
\begin{equation}\label{eq:t-r-mu}
|t|=\frac{2}{\mu+\mu^{-1}}, \quad |r|=\frac{\mu-\mu^{-1}}{\mu+\mu^{-1}}.
\end{equation}
In particular, for an angle to be (approximately) magic, the corresponding real matrix should be an (approximate) rotation. 

Write for a transmission problem through an n-p barrier of characteristic width~$l$ and for the angle of incidence $\theta$
\begin{equation}\label{eq:np-singular}
\tA_{n-p} = \tA_{n-p} (l,\theta) =R_{\alpha(l,\theta)}\left(\begin{smallmatrix} \mu(l,\theta) & 0 \\ 0 & \mu(l,\theta)^{-1} 
\end{smallmatrix} \right) R_{\beta(l,\theta)},
\end{equation}
where $R_{\alpha}$ stays for the rotation at angle~$\alpha$.

Next, consider the p-n transmission problem. To do so, note, that the application of $T$ sends it to the solution of the form
\begin{equation}
\Psi_{p-n}(x)= \begin{cases}
        \overline{a}_1\psiR(x)+\overline{a}_2 \psiL(x)+\overline{a}_3 \psi_-(x),& \\
         \qquad \text{ to the right of the barrier}, & \\
        \overline{a}_4\phiR(x)+\overline{a}_5 \phiL(x)+\overline{a}_6 \phi_+(x), & \\
         \qquad \text{ to the left of the barrier} &
    \end{cases}
\end{equation}
for the new potential $U_{p-n}(x)=U_{n-p}(-x)$ (note that this potential will be supported on $[-l,0]$). For the corresponding transmission matrix, we then have
\begin{equation}\label{eq:inv-conj}
A_{p-n}=\overline{A}_{n-p}^{\,-1}.
\end{equation}
In the same way as before, we pass to the sine-cosine bases, thus obtaining a real determinant 1 matrix
\begin{equation}\label{eq:pn-tilde}
\tA_{p-n} = \frac{\sqrt{-\kn\En}}{\sqrt{kE}} \left(\begin{smallmatrix} \frac{1}{2} & \frac{1}{2} \\ -\frac{1}{2i} & \frac{1}{2i} 
\end{smallmatrix} \right) A_{p-n}  \left(\begin{smallmatrix} \frac{1}{2} & \frac{1}{2} \\ \frac{1}{2i} & -\frac{1}{2i} 
\end{smallmatrix} \right)^{-1}
\end{equation}
Joining~\eqref{eq:np-tilde},~\eqref{eq:inv-conj} and~\eqref{eq:pn-tilde}, and using the reality of matrices $\widetilde{A}_{n-p}$, $\widetilde{A}_{p-n}$, we get 
\begin{equation}\label{eq:duality}
\widetilde{A}_{p-n}=
\left(\begin{smallmatrix} 1 & 0 \\ 0 & -1
\end{smallmatrix} \right)
\widetilde{A}_{n-p}^{-1}
\left(\begin{smallmatrix} 1 & 0 \\ 0 & -1
\end{smallmatrix} \right).
\end{equation}
The singular decomposition~\eqref{eq:np-singular} thus gives us
\begin{multline}\label{eq:pn-singular}
\tA_{p-n} = 
\left(\begin{smallmatrix} 1 & 0 \\ 0 & -1
\end{smallmatrix} \right)
 R_{-\beta(l,\theta)}
 \left(\begin{smallmatrix} \mu(l,\theta) & 0 \\ 0 & \mu(l,\theta)^{-1} 
\end{smallmatrix} \right)^{-1} \times \\ \times R_{-\alpha(l,\theta)} 
\left(\begin{smallmatrix} 1 & 0 \\ 0 & -1
\end{smallmatrix} \right) =  
 R_{\beta(l,\theta)}
 \left(\begin{smallmatrix} \mu(l,\theta)^{-1} & 0 \\ 0 & \mu(l,\theta)
\end{smallmatrix} \right) R_{\alpha(l,\theta)}.
\end{multline}

%
Finally, let us return back to an n-p-n junction. Note, that if the width~$l_2$ of the flat p-part of the barrier is sufficiently large, the transmission through the barrier can be approximated as a sequence of an n-p and p-n junctions (of widths~$l_1$ and~$l_3$ respectively). Indeed, the corresponding physical solutions are of the form
$$
\Psi_{n-p-n}(x)= \begin{cases}
        a_1\psiR(x)+a_2 \psiL(x)+a_{3} \psi_+(x), & \\
         \qquad \text{ to the left of the barrier}, &\\
        a_4\phiR(x)+a_5 \phiL(x)+a_6\phi_-(x)+a_7 \phi_+(x),& \\
         \qquad \text{ in the p-zone}, & \\
        a_8\psiR(x)+a_9 \psiL(x)+a_{10} \psi_-(x), & \\
         \qquad \text{ to the right of the barrier}. &
    \end{cases}
$$
The component $a_7\phi_+(x)$ is of order at most $1$ near the p-n part, and hence is of order $e^{-\lmdn l_2}$ near the n-p transition, which is neglectable (it is less than $10^{-7}$ for the barrier and angles described on Figs.~\ref{fig:bilayer},~\ref{fig:sym-vs-asym}).
The same applies to $a_6\phi_-(x)$. Thus, with very high accuracy one has
$$
\left( \begin{array}{c}
a_1 \\
a_2
\end{array} \right)\approx A_{n-p} \left( \begin{array}{c}
a_4 \\
a_5
\end{array} \right), \quad
\left( \begin{array}{c}
a_4 \\
a_5
\end{array} \right)\approx A_{p-n}^{(l')} \left( \begin{array}{c}
a_8 \\
a_9
\end{array} \right),
$$
where the transmission matrix $A_{p-n}^{(l')}$ corresponds to the p-n barrier, shifted by~$l'=l_1+l_2+l_3$ to the right from the position~$[-l_3,0]$ at which it was studied earlier. For the full transition matrix we thus have
$$
A_{n-p-n}\approx A_{n-p} A_{p-n}^{(l')}.
$$

Hence, the same holds for the transmission matrix in the sine-cosine basis:
\begin{equation}\label{eq:n-p-n}
\widetilde{A}_{n-p-n}\approx \widetilde{A}_{n-p} \widetilde{A}_{p-n}^{(l')}.
\end{equation}
Taking into account that 
$$
\widetilde{A}_{p-n}^{(l')}= R_{-\kn l'} \widetilde{A}_{p-n} R_{kl'}^{-1}
$$
and substituting into~\eqref{eq:n-p-n} the singular decompositions~\eqref{eq:np-singular},~\eqref{eq:pn-singular}, we obtain for the total n-p-n transmission matrix
\begin{multline}\label{eq:A-approx}
\tA_{n-p-n}(\theta)\approx  R_{\alpha_1(\theta)}\left(\begin{smallmatrix} \mu_1(\theta) & 0 \\ 0 & \mu_1^{-1}(\theta) \end{smallmatrix} \right) R_{\beta_1(\theta)} \times \\ \times R_{\beta_2(\theta)-\kn l'}\left(\begin{smallmatrix} \mu_2^{-1}(\theta) & 0 \\ 0 & \mu_2(\theta) \end{smallmatrix} \right) R_{\alpha_2(\theta)-kl'},
\end{multline}
where  
\begin{align*}
\mu_1(\theta)=\mu(\theta,l_1),\quad \mu_2(\theta)=\mu(\theta,l_3), \\
\alpha_1(\theta)=\alpha(\theta,l_1),  \quad \alpha_2(\theta)=\alpha(\theta,l_3), \\
\beta_1(\theta)=\beta(\theta,l_1),  \quad \beta_2(\theta)=\beta(\theta,l_3).
\end{align*}
This matrix is natural to expect to be closest to a rotation for angles $\theta$ when the composition of rotations in the middle, $R_{\beta_1(\theta)+\beta_2(\theta)-\kn l'}$, is a rotation by an integer multiple of $\pi$ (compare with~\eqref{eq:p-trace} and~\eqref{eq:trAA} for in Sec.~\ref{s:peaks} below).
For such angle $\theta_0$, one has
\begin{equation*}
\tA_{n-p-n}(\theta)\approx \pm R_{\alpha_1(\theta)}\left(\begin{smallmatrix} \mu_1(\theta)\mu_2^{-1}(\theta) & 0 \\ 0 & \mu_1^{-1}(\theta)\mu_2(\theta)  \end{smallmatrix} \right) R_{\alpha_2(\theta)-kl'}.
\end{equation*}


Now, a final (and key) remark is that for the bilayer graphene and relatively short n-p barriers, the value $\mu(\theta,l)$ does depend on the angle of incidence, but depends very slightly on the width of the barrier, while the latter stays sufficiently small: see Fig.~\ref{fig:mu-theta}. 
A plausible explanation for this will be discussed in Sec.~\ref{s:algorithm}. 

Due to the this observation, the above matrix is indeed very close to the rotation one, as $\mu(\theta,l_1)\approx \mu(\theta,l_3)$. More precisely, for such $\theta_0$, substituting $\mu=\mu_1(\theta_0)/\mu_2(\theta_0)$ into~\eqref{eq:t-r-mu}, one has
$$
|r(\theta_0)| = \frac{|\mu-\mu^{-1}|}{\mu+\mu^{-1}}=\tanh \delta(\theta_0),
$$
where $\delta(\theta)=|\log \mu_1(\theta)-\log \mu_2(\theta)|$.
This scenario indeed holds for both approximate magic angles illustrated on Fig.~\ref{fig:bilayer}, right. For instance, for an approximate magic angle $\theta_1\approx 30.65^{\circ}$ we have
\begin{align*}
\log \mu_1(\theta_1)\approx 1.1641, \quad \log \mu_2(\theta_1)\approx 1.1343, \\
\delta(\theta_1)= |\log \mu_1(\theta_1)-\log \mu_2(\theta_1) |\approx 0.0298, 
\end{align*}
what is in perfect agreement with $|r(\theta_1)| \approx 0.0299$. For $\theta_2\approx 65.34^{\circ}$ one has 
\begin{align*}
\log \mu_1(\theta_2)\approx 1.5815, \quad \log \mu_2(\theta_2)\approx 1.6339, \\
\delta(\theta_2)= |\log \mu_1(\theta_2)-\log \mu_2(\theta_2) |\approx 0.0524, 
\end{align*}
what is also in perfect agreement with $|r(\theta_2)|\approx 0.0524$.

Finally, the above description also explains why the 30-100-30~nm potential was such a good approximation for the 20-100-40~nm one (see Fig.~\ref{fig:sym-vs-asym}). Indeed, the equation for the magic angle is based on the angles $\beta(\theta,l_1), \beta(\theta,l_3)$, and we have with quite high precision
$$
\beta(\theta,l_1)+\beta(\theta,l_3)\approx 2 \beta(\theta,\frac{l_1+l_3}{2}).
$$

To conclude, we note that the mechanism described in this section is exactly the one that one would like to avoid while looking for a transmission-blocking potential. Hence, it seems quite natural to avoid long flat parts (as they are likely to ``cancel'', for some angles, what happens before and after them), and hence consider a fast-oscillating potential~--- as it was done in Sec.~\ref{s:no-transmission}.

\begin{figure}
\subfloat[]{\includegraphics[width=7cm]{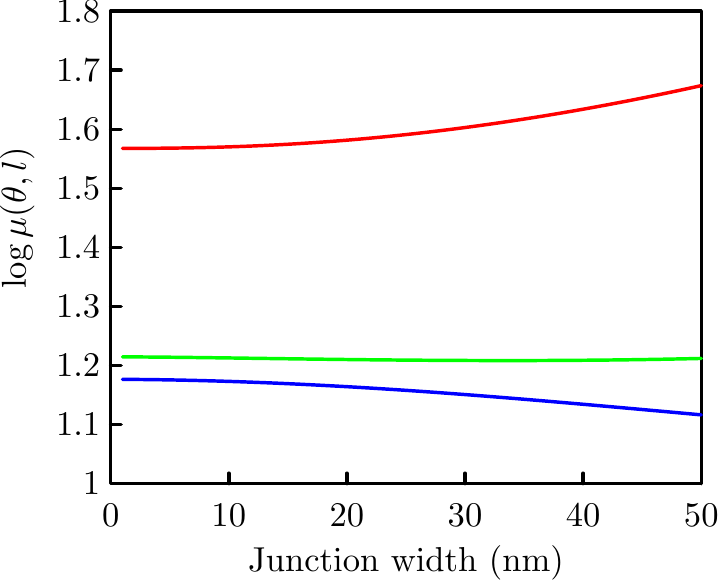}}\\ 
\subfloat[]{\includegraphics[width=7cm]{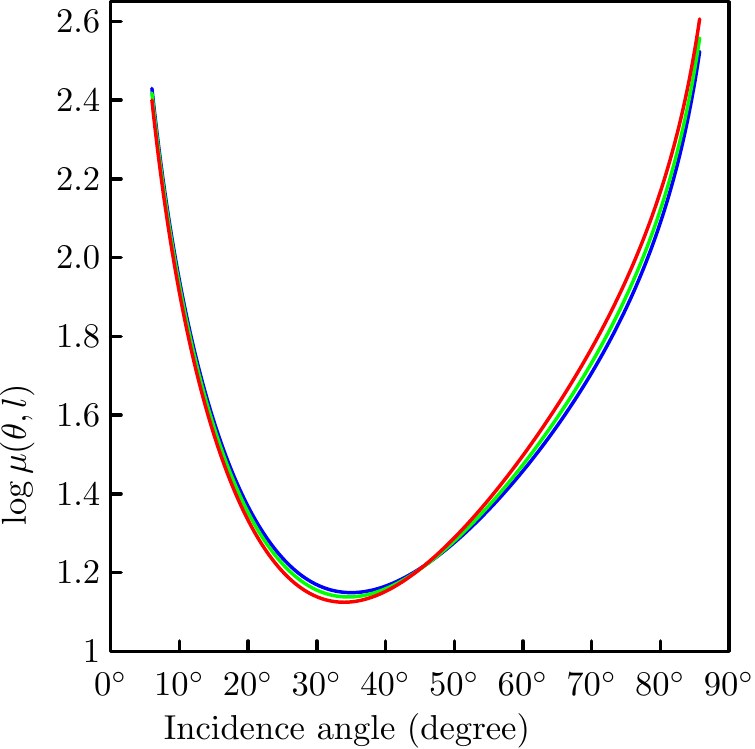}}
\caption{The function $\log \mu(\theta,l)$ for a wave of energy $E=17$~meV, crossing on bilayer graphene an n-p barrier of height $V=50$~meV (same energy and height as for Fig.~\ref{fig:bilayer}).
(a) the dependence on the width $l$ for (almost) magic angles $\theta\approx 30.64^{\circ}$ (blue line), $\theta\approx 65.34^{\circ}$~(red) and for $\theta=45^{\circ}$ (green). (b) the dependence on the angle of incidence~$\theta$, for widths~$l=20$~nm (blue line), $30$~nm (green), $40$~nm (red).
}\label{fig:mu-theta}
\end{figure}

\subsection{Peaks for n-p-n barriers}\label{s:peaks}

The description above can be used to describe the shape of a peak of the transmission probability. Namely, for a matrix $A\in SL(2,\mathbb{R})$, the corresponding inertia coefficients $\mu$ and $\mu^{-1}$ can be found from the relation $\tr AA^* = \mu^2+ \mu^{-2}$, thus implying that 
\begin{equation}\label{eq:p-trace}
p(\theta)=|t(\theta)|^2= \frac{4}{2+\tr AA^*},
\end{equation}
where $A=\tA_{n-p-n}$ is the corresponding transmission matrix. From the approximation~\eqref{eq:A-approx}, one gets 
$$
\tr AA^*=\tr \left(\begin{smallmatrix} \mu_1^2(\theta) & 0 \\ 0 & \mu_1^{-2}(\theta) \end{smallmatrix} \right) R_{\gamma(\theta)}  \left(\begin{smallmatrix} \mu_2^2(\theta) & 0 \\ 0 & \mu_2^{-2}(\theta) \end{smallmatrix} \right) R_{-\gamma(\theta)},
$$
where $\gamma(\theta)=\beta_1(\theta)+\beta_2(\theta)-\kn l'$. An explicit calculation then gives 
\begin{multline}\label{eq:trAA}
\tr AA^* = \left(\mu_1^2(\theta)\mu_2^2(\theta)+\mu_1^{-2}(\theta)\mu_2^{-2}(\theta)\right) \cos^2 \gamma(\theta)+ \\
+\left(\mu_1^2(\theta)\mu_2^{-2}(\theta)+\mu_1^{-2}(\theta)\mu_2^2(\theta)\right) \sin^2 \gamma(\theta)=\\
= c_1(\theta) + c_2(\theta) \cos^2 \gamma(\theta),
\end{multline}
where 
\begin{align*}
c_1(\theta)=\mu_1^2(\theta)\mu_2^{-2}(\theta)+\mu_1^{-2}(\theta)\mu_2^2(\theta), 
\\
c_2(\theta)=(\mu_1^2(\theta)-\mu_1^{-2}(\theta))(\mu_2^2(\theta)-\mu_2^{-2}(\theta)).
\end{align*}

A peak of the transmission probability corresponds to a local minimum of the denominator of~\eqref{eq:p-trace}. Due to the wide $l_2$-part, the angle 
$$
\gamma(\theta)= \beta_1(\theta)+\beta_2(\theta)-\kn l' 
$$ 
changes much faster, than $\mu_{1,2}(\theta)$ do, so such a peak~$\theta_0$ (almost) corresponds to $\gamma(\theta_0)\approx \pi n$. Approximating $\cos^2(\gamma(\theta))\approx c (\theta-\theta_0)^2$, where $c= (\frac{d \gamma (\theta)}{d\theta})^2$, one gets for the shape of the peak the Lorenz--Breit--Wigner form:
$$
p(\theta)=\frac{4}{2+\tr AA^*}\approx \frac{p(\theta_0)}{1+c' (\theta-\theta_0)^2}.
$$
Such an approximation turns out to be quite precise: see Fig.~\ref{fig:quad} for such an approximation for the peaks corresponding to the n-p-n barrier discussed in~\cite{TRK12} (as well as in Sec.~\ref{s:intro-b}).

\begin{figure}[hbt]
    \centering
    \includegraphics[width=8.5cm]{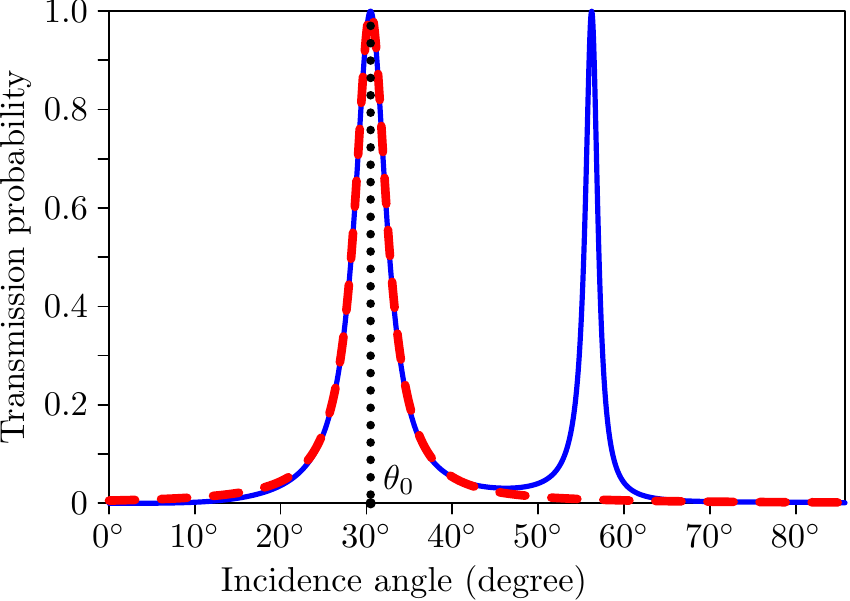}
    \caption{
        {Blue (solid) line: exact transmission probability for the symmetric n-p-n junction. Red (dashed bold) line: Lorenz--Breit--Wigner approximation $p(\theta)\approx \frac{1}{1+c(\theta-\theta_0)^2}$.}
    }\label{fig:quad}
\end{figure}


\subsection{Computational algorithms}\label{s:algorithm}

We conclude this paragraph with specifying a computation method for finding the transmission probabilities for long barriers. Namely, a straightforward method of computation includes numerically solving the differential equation~\eqref{eq:D-x-4x4} through the barrier. Then, finding a linear combination $\sB$ of solutions that starts on the right of the barrier with $\psiR$ and with $\psi_{-}$ that would have no exponentially growing component on the left of the barrier. Though, if the barrier is sufficiently long, the solution starting with a given initial value has a tendency to grow exponentially with the width of the barrier.

Instead, consider the solution $\sM(x)$ to~\eqref{eq:D-x-4x4} that coincides with $\psi_-$ on the right of the barrier. Note, that we do not need to know the solution itself, but only up to the proportionality: its only role will be to be added to a linear combination to remove the exponential growth on the left of the barrier. Hence, instead of finding $\sM(x)$ itself (which is most natural to expect to grow exponentially when one passes to the left of the barrier), we can look for a unit vector-valued function $\Psi_-^0(x)$ that is proportional to it at any point $x$, given by the normalization
\begin{equation}
\Psi_-^0(x)=\sM(x)/|\sM(x)|.
\end{equation}
The latter can be found either by normalizing the solution on each step, or by solving a differential equation for it,
\begin{equation}\label{eq:Psi0}
\partial_x \Psi_-^0 = A\Psi_-^0 - (A \Psi_-^0,\Psi_-^0) \Psi_-^0,
\end{equation}
where $A$ be the matrix of~\eqref{eq:D-x-4x4} and  $(\cdot,\cdot)$ is the usual scalar product.


Second, to find the reflection and transmission coefficients, we are looking for a solution $\sB$ to~\eqref{eq:D-x-4x4} that is of the form~\eqref{eq:t-and-r-bilayer}. Again, it is natural to expect that the solution, starting with $\psi_R$ on the right of the barrier, will grow exponentially. 

But the solution $\sB$ we are looking for is anyway a combination of the solution starting with $\psiR$ and the one starting with~$\psi_{-}$. So while solving the equation~\eqref{eq:D-x-4x4} right-to-left, starting with $\psiR$ outside the barrier, we can at any point $x$ safely add $\Psi_+^0(x)$ with any coefficient: this keeps us in the same space of linear combinations. 
Thus, we consider the component of $\psiR$ that is orthogonal to $\Psi_-^0$, that is, 
\begin{equation}\label{eq:p-projection}
\wpsi(x):=\psi_R(x)-(\psi_R(x),\Psi_-^0(x)) \Psi_-^0(x).
\end{equation}

It obeys the following modification of the equation~\eqref{eq:D-x-4x4}:
$$
\partial_x \wpsi = A\wpsi - \frac{1}{2}((A+A^*)\wpsi,\Psi_-^0) \Psi_-^0.
$$
Such an orthogonalization removes the growth associated to the highest (non-physical) eigenvalue of $A$, leaving only the expansion and contraction associated to the physical solutions themselves.

Finally, after arriving to the point $(-x_0)$ on left side of the barrier, one further modifies the solution $\widetilde{\psi}$ on the left of the barrier by adding again $\Psi_-^0$ with such a coefficient
$$
\beta=-\frac{Q(\widetilde{\psi}(-x_0),v_+)}{Q(\Psi_-^0(-x_0), v_+)}
$$
that the obtained combination $\psi:=\widetilde{\psi}+\beta \Psi_-^0$ does not have a component along $\psi_-$. Decomposing $\psi(-x_0)$ in the base of eigenvectors $v_R, v_L, v_+, v_-$, we find $\frac{1}{t(\theta)}$ and $\frac{r(\theta)}{t(\theta)}$ as coefficients before $v_R$ and $v_L$ respectively.

The above algorithm (or its slight modification) seem also to be an explanation of the almost-constancy of the inertia coefficients $\mu(\theta)$, that are (as it was discussed in Sec.~\ref{s:approximate}) in the background of the appearance of almost-magic angles in the bilayer case.

Note first, that we can construct the space of physically interesting solutions of~\eqref{eq:D-x-4x4} in the following way. Any bounded (``physical'') solution $\widetilde{\Psi}$ of~\eqref{eq:D-x-4x4} is $Q$-orthogonal to both solutions $\sP$ and $\sM$ that coincide respectively with $\psi_+$ to the left of the barrier and with $\psi_-$ to the right of it. Indeed, $Q$-products $Q(\widetilde{\Psi},\sP)$ and $Q(\widetilde{\Psi},\sM)$ give (after normalization by $Q(\psi_-,\psi_+)$) respectively the coefficients before $\psi_-$ in the decomposition to the left of the barrier and before $\psi_+$ in the decomposition to the right of it. On the other hand, for a generic $\theta, E, U$ the solutions $\sP$ and $\sM$ are linearly independent, and hence a $Q$-orthogonal complement to the 2-plane $\{\alpha \sP + \beta \sM\}$ is exactly the 2-plane of physical solutions.

Now, instead of considering functions $\sM$ and $\sP$, we can consider their normalized versions
\begin{equation*}
\Psi_+^0(x)=\sP(x)/|\sP(x)|, \quad \Psi_-^0(x)=\sM(x)/|\sM(x)|.
\end{equation*}
As before, they can be constructed by solving a differential equation left-to-right and right-to-left respectively. Note that it is natural to expect such a construction to be stable under small perturbations: a long-time map contracts almost all the directions in a small neighborhood of the most expanded image.

Once these two functions are found, at any point $x$ we know the 2-plane of physical solutions as the $Q$-orthogonal complement to the 2-plane $\{\alpha \Psi_+^0(x) + \beta \Psi_-^0(x)\}$. 

Finally, knowing such plane, generically, we can reconstruct the flow without using the potential~$U$ explicitly. Indeed, the derivatives of the first and of the third coordinate in~\eqref{eq:D-x-4x4} do not include~$\tilde{U}$, and at the same time, the first and the third coordinate generically provide a system of coordinates on the 2-plane of physical solutions. Thus, the system of differential equations on the first and third coordinate can be written only using~$\Psi_{\pm}^0$.

The same technique applies also to the study of n-p transitions, with the only difference that on the right of the barrier we are taking the function $\Phi_-^0$ instead of $\Psi_-^0$, that is obtained as $\Phi_-^0=\widetilde{\Phi}_-/|\widetilde{\Phi}_-|$, where $\widetilde{\Phi}_-$ coincides with $\phi_-$ on the right of the barrier.

At the same time, it is natural to expect that $\Psi_+^0$ and $\Phi_-^0$ will not change too much under small changes of the potential $U$ defining the n-p transition. The reason for that is that these functions are solutions to 1st order differential equations, defined by $U$ (so even a discontinuity of $U$ will result in a smooth behavior of $\Psi_+^0$ and~$\Phi_-^0$). At the same time, in a zone where two different potentials $U$ coincide, these functions (roughly speaking, corresponding to the fastest-growing direction) quickly become aligned, so the effect of any perturbation should be sufficiently local.

This shows the difference between the single- and bilayer graphene. For a single-layer graphene, we have a two-dimensional differential equation directly involving the potential $U$. While for the bilayer graphene, we have an additional ``integration'': the solutions $\Psi_+^0$, associated to these potentials, are very close to each other  (and the same holds for the solutions $\Phi_-^0$); they mostly coincide, and differ only a little between the two potentials. And as it is these functions that define a new 2-dimensional differential equation, the difference between the transmission matrices $A_{n-p}$ comes from the ``integration'' of difference between them, and should be even smaller.

\section{Conclusions}
To summarize, we have considered chiral tunneling through one-dimensional potential barriers for the cases of both single-layer and bilayer graphene. We have proven that in both cases for symmetric barriers magic angles with 100\% transmission can be found from one {\it real} equation. For the case of bilayer, this equation does not necessarily have real solutions and we have presented examples of the potential barrier with a very low transmission at any angles in a restricted energy range. This opens a way to build the conventional p-n-p (or n-p-n) transistor from bilayer graphene {\it without opening the energy gap}. This is also important conceptually since it gives a clear counterexample to an opinion that the Klein tunneling in single-layer graphene is due to gapless character of the spectrum; actually, it is due to a special chiral character of electronic states there. We have also presented arguments explaining why for the case of bilayer graphene the difference between symmetric and antisymmetric barriers as not as dramatic as for the single-layer one: whereas for the latter case asymmetric shape of the barrier results in an exponential suppression of transmission at non-zero magic angles, in bilayer graphene they are robust, in a sense that the transmission probability remains close to 100\%. 

\section*{Acknowledgements}

MIK acknowledges funding from the European Union Seventh Framework Programme under grant agreement No.~604391 Graphene Flagship and from ERC Advanced Grant No.~338957 FEMTO/NANO.

VK, AO and IS acknowledge the support of the Russian Foundation for Basic Research project \mbox{13-01-00969-a}.

The research of IS is also supported in part by the Dynasty Foundation.

\end{document}